\documentclass[11pt]{article}
\usepackage{picins}
\usepackage{wrapfig}
\usepackage{amsfonts}
\usepackage{amssymb,epic,eepic}
\usepackage{amsmath}
\usepackage{dcolumn}
\usepackage{bm}
\usepackage{bbm}
\usepackage{verbatim}
\usepackage{color}
\usepackage{amsthm}
\usepackage{tikz}
\usetikzlibrary{arrows,positioning,decorations.pathmorphing,decorations.markings,shapes,fadings}
\newcommand{\hr}{{\mathcal H}}

\newcommand{\cs}{{\mathcal S}}

\newcommand{\kr}{{\mathcal K}}

\newcommand{\cc}{{\mathbb C}}

\newcommand{\M}{{\mathcal M}}
\newcommand{\nn}{{\mathbb N}}

\newcommand{\eps}{{\varepsilon}}        

\newcommand{\cP}{\mathcal P}

\newcommand{\bX}{\mathbf X}
\newcommand{\bY}{\mathbf Y}

\newcommand{\bA}{\mathbf A}
\newcommand{\bB}{\mathbf B}
\newcommand{\bC}{\mathbf C}
\newcommand{\eins}{{\mathbbm{1}}}
\newcommand{\bbmY}{\mathbbm Y}

\newtheorem{theorem}{Theorem}

\newtheorem{conjecture}{Conjecture}

\newtheorem{lemma}{Lemma}

\newtheorem{remark}{Remark}

\newcommand{\tr}{\mathrm{tr}}
\newcommand{\supp}{\mathrm{supp}}

\DeclareMathOperator{\sgn}{sgn}
\DeclareMathOperator{\linspan}{span}

\begin{document}
\title{A solution to two party typicality using representation theory of the symmetric group} \author{Janis N\"otzel$^{1,2}$ \\ \scriptsize{Electronic address: janis.noetzel@tum.de, janis.notzel@uab.cat}
\vspace{0.2cm}\\
{\footnotesize $^{1}$ Theoretische Informationstechnik, Technische Universit\"at M\"unchen,}\\
{\footnotesize 80290 M\"unchen, Germany}
\\\\
\footnotesize{$^{2}$F\'{\i}sica Te\`{o}rica: Informaci\'{o} i Fen\`{o}mens Qu\`{a}ntics, Universitat Aut\`{o}noma de Barcelona,}\\
\footnotesize{ES-08193 Bellaterra (Barcelona), Spain}
}

\maketitle

\begin{section}{Abstract}
We give a proof of the multi-party typicality conjecture for the first nontrivial case when there are only two parties. The conjecture itself is motivated by the study of multi-party state merging protocols on quantum systems. Our approach is based on fundamental group-theoretical properties, thereby providing an opportunity to study the problem from a more systematic perspective. Our proof also covers an extended multiparty typicality conjecture that we state in this work. This extended multiparty typicality conjecture is formulated using arbitrary $k$-norms instead of only the $2$-norm as in the original conjecture.
\end{section}
\begin{section}{Preliminaries}
We first introduce and motivate the problem in subsection \ref{subsec:Introduction}, after which we fix some initial notation in \ref{subsec:Basic Notation}. Then, we proceed with a reformulation and generalization of the initial conjecture in subsection \ref{subsec:Reformulation of the conjecture and initial approaches}. We explain how a typical approach to prove validity of the conjecture may look like (it has not been proven yet that such approach would fail, nor has it been demonstrated to succeed), before we end the introductory part by introducing the necessary representation-theoretic language that allows for an elegant solution of the two-party case.
\begin{subsection}{Introduction and motivation of the problem statement\label{subsec:Introduction}}
This work is motivated by a conjecture \cite[Conjecture 3.2.7]{dutil-thesis} that was formulated by Nicolas Dutil. The conjecture is motivated by the observation that certain multi-party protocols on quantum systems require the use of time sharing. It is not within the scope of this rather technical contribution to rewrite the history of multi-party state merging and related protocols. Those readers with an interest in the origins of the conjecture that motivated our work are encouraged to pick up the information directly at the source \cite{dutil-thesis}, or in publications such as \cite{fawzi} or \cite{savov}, where other forms of multiparty-typicality are formulated and discussed. In addition to that, we would like to point the reader interested in one-shot formulations to the work \cite{drescher-fawzi}.\\
In order to give a self-contained approach to the question, we will here take the approach of comparing the asymptotic properties of multi-party i.i.d. probability distributions, when the number of copies goes to infinity, with the corresponding situation of multi-party i.i.d. quantum states. From our comparison, it will become clear that multi-party quantum states are in fact potentially missing one property, and the search for this missing property then serves as the starting point of our analysis. We will confine ourselves to the study of no more than three parties, since three is already enough the number of systems where a separation between probability distributions and quantum states can be observed.\\
Let $\bA$, $\bB$ and $\bC$ be finite sets. Let $p$ be a probability distribution on their cartesian product $\bA\times\bB\times\bC$, meaning that $\sum_{(a,b,c)\in\bA\times\bB\times\bC}p(a,b,c)=1$ and $p(a,b,c)\geq0$ for all $(a,b,c)\in\bA\times\bB\times\bC$. From $p$, we can form all its marginal distributions, like for example the distribution $p_\bA$ on $\bA$ defined by setting for all $a\in\bA$ $p_\bA(a):=\sum_{(b,c)\in\bB\times\bC}p(a,b,c)$ and $p_{\bA\bB}(a,b):=\sum_{c\in\bC}p(a,b,c)$ for all $a\in\bA$ and $b\in\bB$. These definitions extend to $p_{\bB}$, $p_\bC$, $p_{\bB\bC}$ or $p_{\bA\bC}$ in a straightforward fashion.\\
To any of these distributions (we define this only for $p$) and any natural number $n$ we can then define the probability distribution $p^{\otimes n}$ on the $n$-fold cartesian product $(\bA\bB\bC)^n$ by $p^{\otimes n}(a^n,b^n,c^n):=\prod_{i=1}^np(a_i,b_i,c_i)$. We can also define, for every $\delta>0$, typical sets
\begin{align}
T_{p,\delta}:=\left\{(a^n,b^n,c^n): \forall(a,b,c)\ \begin{array}{ll}|n^{-1}\cdot N(a,b,c|a^n,b^n,c^n)-p(a,b,c)|\leq\delta\\ p(a,b,c)=0\ \Rightarrow\ N(a,b,c|a^n,b^n,c^n)=0\end{array}\right\},
\end{align}
where $N(a,b,c|a^n,b^n,c^n)$ is the number of times the symbol $(a,b,c)$ appears in the string $(a^n,b^n,c^n)$ and it is understood that all triples $(a,b,c)$ are elements taken from $\in\bA\times\bB\times\bC$. This definition extends to all the marginal distributions, so that we obtain sets such as $T_{p_\bA,\delta}$ or $T_{p_{\bB\bC},\delta}$ and the like.\\
Let $T_{\delta}$ be the set of all $(a^n,b^n,c^n)$ such that $(a^n,b^n,c^n)\in T_{p,\delta}$, $a^n\in T_{p_\bA,\delta}$, $(a^n,c^n)\in T_{p_{\bA\bC,\delta}}$ and so on and so forth for all the possible marginal distributions of $p$. Let
\begin{align}
q^n:=\|p^{\otimes n}\cdot\eins_{T_{\delta}}\|_1^{-1}\cdot p^{\otimes n}\cdot\eins_{T_{\delta}}.
\end{align}
It then holds that
\begin{align}
\|q^n-p^{\otimes n}\|_1\underset{n\to\infty}{\longrightarrow}0,
\end{align}
where $\eins_{T_{p,\delta}}$ denotes the indicator function taking the value $1$ on $T_{\delta}$ and zero else and $\|\cdot\|_1$ is the usual one-norm. Moreover, the distributions $q^n$ ($n\in\nn$) have the property that all the marginal distributions arising from it obey the estimates
\begin{align}\label{eqn:basic-classical-estimate}
\forall\ \mathcal T\subset\{\bA,\bB,\bC\}:\qquad\|q^n_{\mathcal T}\|_2^2\leq2^{-n(H(p_{\mathcal T})-\gamma(\delta,n))},
\end{align}
for all $n\geq N$ for some appropriately chosen (and large enough) $N\in\nn$ and a function $\gamma:\mathbb R_+\times\nn\to\mathbb R_+$ satisfying $\lim_{n\to\infty}\gamma(\delta,n)=0$ for all $\delta$. Here $H$ denotes the Shannon-entropy which is defined by $H(q):=-\sum_{x\in\bX}q(x)\log q(x)$, for arbitrary alphabets $\bX$ and probability distributions $q$ on them, and $\|\cdot\|_2$ is the usual two-norm.\\
Thus, it is possible to find an approximation to $p^{\otimes n}$ that not only approximates $p^{\otimes n}$ asymptotically perfect (note that this implies the same for all the corresponding marginal distributions, since $\|\cdot\|_1$ is monotone under stochastic maps) but also delivers a second (and, actually, up to $k$-th order for any fixed $k\in\nn$) order asymptotic scaling that is dictated by information-theoretically relevant functions.\\
It is a natural question to ask for a similar result for quantum states, and this question is the content of the multiparty typicality conjecture, that we reformulate here for three parties as follows:
\begin{conjecture}[Multiparty typicality conjecture - Conjecture 3.2.7 in \cite{dutil-thesis}]\label{conjecture-1} Consider $n$ copies of an arbitrary multiparty state $\rho_{ABC}$. For any fixed $\eps>0$, $\delta_\mathcal{T}>0$ and $n$ large enough, there exists a state $\Phi_{ABC}$ which satisfies
\begin{align}
\|\Phi_{ABC}-\rho_{ABC}^{\otimes n}\|_1&\leq\nu(\eps)\\
\|\Phi_{\mathcal T}\|_2 ^2&\leq(1-\mu(\eps))2^{-n(S(\rho_\mathcal{T})-\delta_{\mathcal{T}})}
\end{align}
for all non-empty subsets $\mathcal T\subset\{A,B,C\}$. Here, $\nu(\eps)$ and $\mu(\eps)$ are functions of $\eps$ which vanish by choosing arbitrarily small values for $\eps$.
\end{conjecture}
Before we come to a more detailed discussion of the conjecture, we first fix some of the notation that is necessary for an understanding of the topic, and of above conjecture.
\end{subsection}
\begin{subsection}{Basic Notation\label{subsec:Basic Notation}}
All Hilbert spaces are assumed to have finite dimensions and are over the field $\cc$. The set of linear operators from $\hr$ to $\hr$ is denoted $\mathcal B(\hr)$. The adjoint of $b\in\mathcal B(\hr)$ is written $b^\dag$.\\
$\cs(\hr)$ is the set of states, i.e. positive semi-definite operators with trace (the trace function on $\mathcal B(\hr)$ is written as $\tr$) $1$ acting on the Hilbert space $\hr$. Pure states are given by projections onto one-dimensional subspaces. A vector $x\in\hr$ of length one spanning such a subspace will therefore be referred to as a state vector, the corresponding state will be written as
$|x\rangle\langle x|$.\\
The von Neumann entropy of a state $\rho\in\mathcal{S}(\hr)$ is given by
\begin{equation}S(\rho):=-\textrm{tr}(\rho \log\rho),\end{equation}
where $\log(\cdot)$ denotes the base two logarithm which is used throughout the paper.\\
Given two states $\rho,\sigma\in\cs(\mathbb C^d)$, the relative entropy of them is defined as
\begin{align}
D(\rho\|\sigma):=\left\{\begin{array}{l l}\tr\{\rho(\log(\rho)-\log(\sigma)\},&\mathrm{if}\ \supp(\rho)\subset\supp(\sigma),\\\infty,&\mathrm{else} \end{array}\right.
\end{align}
Another way of measuring distance between quantum states is obviously given by using the one-norm $\|\cdot\|_1$, which is defined via setting
\begin{align}
\|A\|_1:=\tr\left\{\sqrt{A^\dag A}\right\}\qquad \forall A\in\mathcal B(\hr).
\end{align}
Other well-known norms on operator spaces that need to be defined for an understanding of the topic are the $k$-norms, for arbitrary $k\in\mathbb N$:
\begin{align}
\|A\|_k:=\left(\tr\left\{(A^\dag A)^{k/2}\right\}\right)^{1/k}\qquad \forall A\in\mathcal B(\hr).
\end{align}
In order to understand the multiparty-typicality conjecture we additionally need to define marginal states of multiparty-states. This is done by first introducing the partial trace $\tr_B$. Given two Hilbert-spaces $\hr_A$ and $\hr_B$, this is a map $\tr_B:\mathcal B(\hr_A\otimes\hr_B)\to\mathcal B(\hr_A)$ is the unique map such that for all operators $X\in\mathcal B(\hr_A\otimes\hr_B)$ of the form $X=F\otimes G$ for some $F\in\mathcal B(\hr_A)$ and $G\in\mathcal B(\hr_B)$ we have $\tr_B(X)=F$.\\
Letting now $\rho\in\cs(\hr_A\otimes\hr_B)$, we can define its marginal density operators $\rho_A\in\cs(\hr_A)$ and $\rho_B\in\cs(\hr_B)$ via $\rho_A:=\tr_B\{\rho_{AB}\}$ and $\rho_B:=\tr_A\{\rho_{AB}\}$. For a given $\rho_{AB}\in\cs(\hr_A\otimes\hr_B)$ we denote its spectrum (the ordered lists of its eigenvalues, starting with the largest one, counting multiplicities) as $r_{AB}$. The spectra of $\rho_A$ and $\rho_B$ are denoted $r_A$ and $r_B$, respectively.

For a finite set $\mathbf X$ the notation $\cP(\mathbf X)$ is reserved for the set of probability distributions on $\mathbf X$, and $|\mathbf X|$ denotes its cardinality. Given two alphabets $\bX$ and $\bY$ we will sometimes denote elements of $\cP(\bX\times\bY)$ by e.g. $p_{\bX\bY}$, and in that case it is understood that $p_\bX\in\cP(\bX)$ and $p_\bY\in\cP(\bY)$ denote the respective marginal distributions of $p_{\bX\bY}$. For any $n\in\nn$, we define $\bX^n:=\{(x_1,\ldots,x_n):x_i\in\bX\ \forall i\in\{1,\ldots,n\}\}$, we also write $x^n$ for the elements of $\bX^n$. Given such element, $N(\cdot|x^n)$ denotes its type, and is defined through $N(x|x^n):=|\{i:x_i=x\}|$. The notion of type is actually slightly more general than that, as any function $N:\bX\to\nn$ satisfying $\sum_{x\in\bX}N(x)=n$ can be seen as the type of some $x^n\in\bX^n$. Therefore, we will make a slightly more general use of the term ``type'' and use it to denote any function $N:\bX\to\nn$. If the number $n=\sum_{x\in\bX}N(x)$ needs to be specified we will speak of an $n$-type. Every $n$-type naturally defines a set $T_N\subset\bX^n$ through $T_N:=\{x^n\in\bX^n:N(\cdot|x^n)=N(\cdot)\}$. Normalized types are defined as $\bar N(x|x^n):=\tfrac{1}{n}N(x|x^n)$ for all $x^n\in\bX^n$ and $x\in\bX$. For any natural number $n\in\nn$, the notion of type defines a subset $\cP_0^n(\bX)\subset\cP(\bX)$ via $\cP_0^n(\bX):=\{\bar N(\cdot|x^n):x^n\in\bX^n\}$.\\
The complement of $\mathbf D\subset\bX$ within $\bX$ is denoted $\mathbf D^\complement$.\\
For any natural number $L$, we define $[L]$ to be the shortcut for the set $\{1,...,L\}$.
\end{subsection}
\begin{subsection}{Reformulation of the conjecture and initial approaches\label{subsec:Reformulation of the conjecture and initial approaches}}
With above additional structure and keeping in mind that the estimate (\ref{eqn:basic-classical-estimate}) holds true in a more general sense for any of the $k$-norms, one feels tempted to reformulate and extend Conjecture \ref{conjecture-1} to
\begin{conjecture}[Extended multiparty typicality conjecture]\label{conjecture-2} Let $\hr_A$, $\hr_B$ and $\hr_C$ be Hilbert spaces. Let $\rho_{ABC}\in\cs(\hr_A\otimes\hr_B\otimes\hr_C)$. There is a sequence $(\Phi_{ABC}^n)_{n\in\nn}$ of quantum states satisfying $\Phi_{ABC}^n\in\cs((\hr_A\otimes\hr_B\otimes\hr_C)^{\otimes n})$ for all $n\in\nn$ such that for every natural number $k\geq2$ the following holds true: There is a function $\gamma:\mathbb R_+\times\nn\to\mathbb R_+$ such that for all $n\in\nn$ and for all non-empty sets $\mathcal T\subset\{A,B,C\}$,
\begin{align}
\|\Phi_{ABC}^n-\rho_{ABC}^{\otimes n}\|_1&\leq\gamma(\eps,n)\\
\|\tr_{\tau^\complement}\{\Phi_{ABC}^n\}\|_k ^k&\leq2^{-n\cdot(k-1)\cdot(S(\rho_\mathcal{T})-\gamma(\eps,n))}
\end{align}
hold true. Moreover, for all $k\geq2$ it holds that $\lim_{n\to\infty}\gamma(\eps,n)=0$ for all $\eps>0$.
\end{conjecture}
It can easily checked that validity of above conjecture would imply validity of Conjecture \ref{conjecture-1}. The main obstacle one is confronted with here is how to make a guess for the approximating state $\Phi_{ABC}^n$. Constructions that are straightforward generalizations of the one employed in our introductory discussion for probability distributions on finite alphabets do not directly translate to the problem at hand.\\
A typical construction in the quantum case would involve the use of what is called ``frequency-typical subspaces'' (see for example \cite{wilde-book} for precise formulations). We give an exemplary introduction to the topic, thereby concentrating on the two-party case. Any state $\rho_{AB}$ can, upon a choice of the right basis, be written as $\rho_{AB}=\sum_ir_{ABC}(i)|e_i\rangle\langle e_i|$, where $r_{ABC}(i)$ are the singular values of $\rho_{ABC}$ and $|e_i\rangle\langle e_i|$ mutually orthogonal rank-one projections. For a ``frequency'' or ``type'' (a nonnegative function $t:\{1,\ldots,d\}\to\nn$ satisfying $\sum_it(i)=n$), define the frequency typical subspaces
\begin{align}
V_t:=\linspan(\{e_{i_1}\otimes\ldots\otimes e_{i_n}:|\{k:i_k=i\}|=t(i)\ \mathrm{for\ all\ }i\}.
\end{align}
The corresponding orthogonal projections $P_{V_t}$ onto these subspaces are the frequency-typical subspaces of $\rho_{AB}$. Fixing a $\delta>0$, one could now define $P_\delta:=\sum_{\|t-r_{AB}\|_1\leq\delta}P_t$, trying to reproduce the known approach that we outlined in the introduction. This would lead to the definition
\begin{align}
\Psi_{AB}^n:=(\tr\{P_\delta\rho_{AB}^{\otimes n}\})^{-1}P_\delta\rho_{AB}^{\otimes n}P_\delta.
\end{align}
However, no direct method has so far been demonstrated to yield the desired bounds for this state, the only approach \cite{dutil-thesis} that is known to the author uses a more complex approach.\\
In contrast to that, our approach is able to fully satisfy the classical intuition gained from our introduction, albeit only for two parties.
\end{subsection}
\begin{subsection}{Notation for representation theoretic objects\label{subsec:Notation for representation theoretic objects}} The symbols $\lambda,\lambda', \nu,\nu', \mu,\mu'$ will be used to denote Young frames. The set of Young frames with at most $d\in\nn$ rows and $n\in\nn$ boxes is denoted $\bbmY_{d,n}$.\\ For a Young Tableau $T$, we write $T_{ij}$ for the entry of $T$ in the $i$-th row and $j$-th column.\\ In the remainder, $\hr_A,\hr_B,\hr$ denote Hilbert spaces with dimensions $d_A,d_B,d$. The numbers $d_A,d_B$ will be arbitrary but constant, while $d$ serves as a ``dummy''-dimension for intermediate statements. Dimensions will also be assumed to be strictly larger than one, since otherwise the statements made in this work become trivial.\\ The symbol $\mathbb B^{A\otimes B}$ denotes the product representation of $S_n\times S_n$ on $\hr_A^{\otimes n}\otimes\hr_B^{\otimes n}$ induced by the standard representations $\mathbb B^A,\mathbb B^B$ of the symmetric group $S_n$ on $\hr_A^{\otimes n}$ and $\hr_B^{\otimes n}$.\\ A representation $\mathbb B^{AB}$ of $S_n$ on $\hr_A^{\otimes n}\otimes\hr_B^{\otimes n}$ is then given by the obvious reordering of the standard representation of $S_n$ on $(\hr_A\otimes\hr_B)^{\otimes n}$. It holds
\begin{align}
\mathbb B^{AB}(\sigma)=\mathbb B^A(\sigma)\otimes\mathbb B^B(\sigma)\ \mathrm{for\ all}\ \sigma\in S_n.\label{eqn13}
\end{align}
The unique complex vector space carrying the irreducible representation of $S_n$ corresponding to a Young Tableau $\lambda$ will be written $F_\lambda$.\\ The multiplicity of an irreducible subspace of $\mathbb B^{X}$ (where $X\in\{A,B,AB,A\otimes B\}$) corresponding to a Young frame $\lambda$ is denoted $m_\lambda^X$.\\
Projections onto the irreducible subspaces of $\mathbb B^{AB}$ are denoted by $P_{\lambda,k}^{AB}$ ($\lambda\in \bbmY_{d_Ad_B,n},\ k\in[m^{AB}_\lambda]$). Implicit here is the choice of a specific set of these, and this set is chosen such that every two different projections are orthogonal (this may be seen as a specific choice of bases for the invariant subspaces $U_\lambda,\ \lambda\in \bbmY_{d_Ad_B,n}$, of the reordering of the standard representation $U\mapsto U^{\otimes n}$ of the unitary group on $(\hr_A\otimes\hr_B)^{\otimes n}$). Another constraint will be given by equation (\ref{eqn12}). Accordingly, projections onto irreducible subspaces of $\mathbb B^{A\otimes B}$ get labelled $P_{\mu,i}^A\otimes P_{\nu,j}^B$ ($\mu,\nu\in \bbmY_{d_A,n},\bbmY_{d_B,n},\ i,j\in[m_\mu^A],[m_\nu^B]$).\\ Whenever it feels right, the superscripts $A,B,AB$ will be omitted. To make up for that, in this case, the symbols $\lambda,\lambda'$ will only be used for projections on $AB$, while $\mu,\mu'$ indicate that a projection on $A$ is being used and $\nu,\nu'$ are only subscripts for projections on the $B$-part.\\ Define, for arbitrary $\mu\in \bbmY_{d_A,n},\ \nu\in \bbmY_{d_B,n},\ \lambda\in \bbmY_{d_Ad_B,n}$ the projections \begin{equation}
 P^A_{\mu}:=\sum_{i=1}^{m_\mu^A}P^A_{\mu,i},\qquad P^B_{\nu}:=\sum_{j=1}^{m_\nu^B}P^A_{\nu,j},\qquad P^{AB}_{\lambda}:=\sum_{k=1}^{m_\lambda^{AB}}P^{AB}_{\lambda,k}.
\end{equation}
The choice we just made for the set $\{P^{AB}_{\lambda,i}:\lambda\in \bbmY_{d,n},\ i\in[m_\lambda]\}$ gets a little more specific now:\\ We will choose these projections such that each $P^A_{\mu}\otimes P^B_{\nu}$ (note that these projections correspond to subspaces which are only \emph{invariant} under the action of $\mathbb B^{AB}$) can, by choosing an appropriate set $\M$, be written as
\begin{align} P^A_{\mu}\otimes P^B_{\nu}=\sum_{(\lambda,i)\in\M}P_{\lambda,i}.\label{eqn12}
\end{align} This is possible due to equation (\ref{eqn13}). Conversely, it implies that each $P_{\lambda,i}$ obeyes the inequality \begin{align} P_{\lambda,i}\leq P^A_{\mu}\otimes P_{\nu}^B \end{align} for exactly one specific choice of $\mu,\nu\in \bbmY_{d_A,n},\bbmY_{d_B,n}$.\\ The set of states on a Hilbert space $\hr$ is written $\cs(\hr)$. The set of probability distributions on a finite set $\bX$ is denoted $\mathcal P(\bX)$, the cardinality of $\bX$ by $|\bX|$.\\ For $\lambda\in \bbmY_{d,n}$, $\bar\lambda\in\cP([d])$ is defined by $\bar\lambda(i):=\lambda_i/n$. If $\rho\in\cs(\hr)$ with $\dim\hr=d$ has spectrum $s\in\cP([d])$, then it will always be assumed that $s(1)\geq\ldots\geq s(d)$ holds and the distance between a spectrum $s$ and a Young frame $\lambda\in \bbmY_{d,n}$ is measured by $\|\bar\lambda-s\|_1:=\sum_{i=1}^d|\bar\lambda(i)-s(i)|$. For two functions $t,t':[d]\to\nn$ we write $t\preceq t'$ if $\sum_{i=1}^kt'(i)\geq\sum_{i=1}^kt(i)$ holds for all $k=1,\ldots,d$.\\
We now define two important entropic quantities, both of which use the base two logarithm. Throughout this work, this function will be written $\log$. Given a finite set $\bX$ and two probability distributions $r,s\in\cP(\bX)$, we define the relative entropy $D(r||s)$ by \begin{align} D(r||s):=\left\{\begin{array}{ll}\sum_{x\in \bX}r(x)\log(r(x)/s(x)),&\mathrm{if}\ s\gg r\\ \infty,&\mathrm{else}\end{array}\right. \end{align} In case that $D(r||s)=\infty$, for a positive number $a>0$, we use the convention $2^{-aD(r||s)}=0$. The relative entropy is connected to $\|\cdot\|$ by the Pinsker's inequality $D(r||s)\geq\frac{1}{2\ln(2)}\|r-s\|^2$. The entropy of $r\in\cP(\bX)$ is defined by the formula \begin{align} H(r):=-\sum_{x\in \bX}r(x)\log(r(x)). \end{align}
\end{subsection}
\end{section}
\begin{section}{Result}
As our main result, we prove the extended multiparty typicality conjecture for two parties, thereby automatically including the original case for two parties. Our result is based on the following sequence of approximating states: For a given $\rho_{AB}$, $\epsilon>0$ and $n\in\nn$ we set
\begin{align}
\Phi_{AB}^n:=\left(\tr\{\left(P_\eps^A \otimes P_\eps^B\right)P_\eps^{AB}\rho_{AB}^{\otimes n}\}\right)^{-1}\cdot\left(P_\eps^A \otimes P_\eps^B\right)P_\eps^{AB}\rho_{AB}^{\otimes n}\left(P_\eps^A \otimes P_\eps^B\right).
\end{align}
Here, the operators $P$ are projections onto suitable representations of $\mathbb B^{A\otimes B}$ and $\mathbb B^{AB}$. A proper definition is given within the next lines. The state $\Phi_{AB}^n$ has the obvious marginal states $\Phi_A^n:=\tr_{\hr_B^{\otimes n}}\{\Phi_{AB}^n\}$ and $\Phi_B^n:=\tr_{\hr_A^{\otimes n}}\{\Phi_{AB}^n\}$.\\
As an additional result, we also give the corresponding lower bounds in Theorem \ref{thm:additionaltheorem}.\\
The necessary estimates for the sequence $(\Phi_{AB}^n)_{n\in\nn}$ to fulfill the requirements of the extended multiparty typicality conjecture are given in inequalities (\ref{eqn-thm:thetheorem-8}), (\ref{eqn-thm:thetheorem-4}) and the right hand inequality of (\ref{eqn-thm:thetheorem-1}). The proof of these inequalities is almost trivial.\\
The remaining inequalities are stated only for sake of completeness, although especially the left hand inequality in (\ref{eqn-thm:thetheorem-1}) is comparably hard to prove. In order to state the theorem, we need to define ``cutted'' $\epsilon$-balls $U_\eps(r)$ of width $\eps$ around a distribution $r\in\cP([d])$ as follows: First, take the usual $\mathfrak B_\eps(r):=\{t\in\cP([d]):\|r-t\|_1\leq\epsilon\}$. Then, set $\mathfrak A(r):=\{t\in\cP([d]):r(i)=0\Rightarrow t(i)=0\ \forall i\in[d]\}$. Finally, define the cutted ball as
\begin{align}
\mathfrak C_\eps(r):=\mathfrak A(r)\cap\mathfrak B_\eps(r),
\end{align}
and for every $n\in\nn$ we will use the additional definition
\begin{align}
\mathfrak C_\eps^n(r):=\{\lambda\in\bbmY_{d,n}:\bar\lambda\in\mathfrak C_\eps(r)\}.
\end{align}
With the use of these cutted balls we can define specific projections as follows: Let $\rho_{AB}\in\cs(\hr_A\otimes\hr_B)$ have spectrum $r_{AB}$ and marginals $\rho_A,\rho_B$ with corresponding spectra $r_A,r_B$. For every $\eps>0$ and $n\in\mathbb N$, define the projections
\begin{align}
P^A_\eps&:=\sum_{\mu\in\mathfrak C_\eps^n(r_A)}P^A_{\mu}\qquad\in\mathcal B(\hr_A^{\otimes n}),\\
P^B_\eps&:=\sum_{\nu\in\mathfrak C_\eps^n(r_B)}P^B_{\nu}\qquad \in B(\hr_B^{\otimes n}),\\
P_\eps^{AB}&:=\sum_{\lambda\in\mathfrak C_\eps^n(r_{AB})}P^{AB}_{\lambda}\qquad\in\mathcal B(\hr_{AB}^{\otimes n}).
\end{align}
The dependence of the projections onto the parameter $n$ will, here and in the following, be suppressed in order to enhance readability. Further, it is understood that $\lambda$, $\mu$ and $\nu$ are Young frames taken from $\bbmY_{d_A,n}$, $\bbmY_{d_B,n}$ and $\bbmY_{d_Ad_B,n}$, repsectively. We are ready to formulate our main theorem:
\begin{theorem}\label{thm:thetheorem}
Let $\rho_{AB}\in\cs(\hr_A\otimes\hr_B)$ have spectrum $r_{AB}$ and marginals $\rho_A,\rho_B$ with corresponding spectra $r_A,r_B$. For every $n\in\nn$, set
\begin{align}\label{eqn:fundamental-definition}
\Phi_{AB}^n:=\left(\tr\{\left(P_\eps^A \otimes P_\eps^B\right)P_\eps^{AB}\rho_{AB}^{\otimes n}\}\right)^{-1}\cdot\left(P_\eps^A \otimes P_\eps^B\right)P_\eps^{AB}\rho_{AB}^{\otimes n}\left(P_\eps^A \otimes P_\eps^B\right).
\end{align}
For every natural number $k\geq2$ there is a function $\gamma:\mathbb R_+\times\nn\mapsto\mathbb R_+$ with the property that, for every $\epsilon>0$, $\lim_{n\to\infty}\gamma(\epsilon,n)=0$ and an absolute constant $c$ such that for all $n\in\nn$ we have
\begin{align}
\label{eqn-thm:thetheorem-4}\|\Phi_{AB}^n-\rho_{AB}^{\otimes n}\|_1&\leq2^{-n\cdot(c\cdot\epsilon^2-\gamma(\eps,n))}\\
\label{eqn-thm:thetheorem-8}\tr\{\left(\Phi_{A}^n\right)^k\}&\leq2^{-n\cdot(k-1)\cdot(H(r_A)-\gamma(\eps,n))}\\
\label{eqn-thm:thetheorem-9}\tr\{\left(\Phi_{B}^n\right)^k\}&\leq2^{-n\cdot(k-1)\cdot(H(r_B)-\gamma(\eps,n))}\\
\label{eqn-thm:thetheorem-1}\tr\{\left(\Phi_{AB}^n\right)^k\}&\leq2^{-n\cdot(k-1)\cdot(H(r_{AB})-\gamma(\eps,n))}.
\end{align}
\end{theorem}
We note that the function $\gamma$ does in addition depend on the Hilbert space dimensions $d_A$ and $d_B$ and on the minimal nonzero eigenvalues of $\rho_{AB}$, $\rho_A$ and $\rho_B$. Exact dependencies can be extracted from the proof, for example inequality (\ref{eqn:s_min}) introduces the dependence between $\gamma$ and $\rho_{AB}$. Moreover, it holds that $c=(4\ln2)^{-1}$.\\
During proofs we will use various approximation techniques, some of which are only valid when $\eps\in(0,1/2)$. The resulting estimates are then collected to produce the functions $\gamma$. It is understood that $\gamma(\eps,n)=1$ whenever $\eps\geq1/2$. The same applies to the functions $\nu_k$ in our next theorem:
\begin{theorem}\label{thm:additionaltheorem}
Under the preliminaries of Theorem \ref{thm:thetheorem}, there exists for every $k\geq2$ a function $\nu_k:\mathbb R_+\times\nn\to\mathbb R_+$ satisfying $\lim_{n\to\infty}\nu_k(\eps,n)=0$ for all $\eps>0$ such that
\begin{align}
\label{eqn:additionaltheorem-1}\tr\{\left(\Phi_{AB}^n\right)^k\}&\geq2^{-n\cdot(k-1)\cdot(H(r_{AB})-\nu_k(\eps,n))}.
\end{align}
\end{theorem}
\begin{remark}\label{remark-0}
This second theorem, albeit very similar in nature to the first, does already give a hint concerning the complexity of the estimates of multiparty states like the one constructed here. Compared to the proof of Theorem \ref{thm:thetheorem}, the proof of Theorem \ref{thm:additionaltheorem} is rather involved. Moreover, we were not able to add to it the obvious lower bounds on $\tr\{(\Phi_A^n)^k\}$ and $\tr\{(\Phi_B^n)^k\}$.\\
Following the ideas presented in the proof of Theorem \ref{thm:additionaltheorem} leads one to consider estimates of the form $\tr\{(\rho_A^{(k-1)})^{\otimes n}\tr_{\hr_B^{\otimes n}}\{(\eins_{\hr_A}^{\otimes n}\otimes P_\nu^B)\rho_{AB}^{\otimes n}P_\lambda^{AB}(\eins_{\hr_B}^{\otimes n}\otimes P_\nu^B)\}\}$ in terms of a deviation of $\bar\nu$ and $\bar\lambda$ from $r_A$ and $r_{AB}$. A detailed analysis will be the topic of future work.
\end{remark}
\begin{remark}
\label{remark-1} Note that $P_\delta^{AB}$ commutes with $P^A_\eps\otimes P_\eps^B$ as well as with $\rho_{AB}^{\otimes n}$. While this may seem too obvious to be stated here, it has a grave impact on the matter: Our proofs make heavy use of this fact. In addition, the straightforward extension of our approach to three parties would make use of projections of the form $P_\lambda^{AB}\otimes P^C_\mu$, and these do in general not commute with, for example, projections of the form $P^A_\nu\otimes P^{BC}_\xi$.\\
This simple observation already completely explains where our approach fails. In addition to that, it gives an idea as to what relations may provide useful objects of study in future research.
\end{remark}
We will need a few preliminary results before proving these theorems. First, a few estimates are needed:\\ With $h(i,j)$ denoting Hook-lenghts (see e.g. \cite{sternberg} for a definition of these combinatorial quantities), the dimensions of the irreducible subspaces of any representation of $S_n$ on $(\mathbb C^d)^{\otimes n}$ ($d>0$) obey the following estimates.
\begin{eqnarray} \frac{n!}{\prod_{i=1}^n(\lambda_i+d+1)!}\leq\frac{n!}{\prod_{(i,j)\in\lambda}h(i,j)}=\dim F_\lambda\leq\frac{n!}{\prod_{i=1}^d\lambda_i!}\qquad (\lambda\in \bbmY_{d,n}).\label{eqn1} \end{eqnarray} Also, we are going to employ the following estimate taken from \cite{csiszar-koerner}, Lemma 2.3: \begin{equation} \frac{1}{(n+1)^d}2^{nH(\overline{\lambda})}\leq\frac{n!}{\prod_{i=1}^d\lambda_i!}\leq2^{nH(\overline{\lambda})}\qquad (\lambda\in \bbmY_{d,n})\label{eqn2} \end{equation} as well as, with $\overline{\lambda+d+1}(i):=\frac{1}{n+d(d+1)}(\lambda_i+d+1)$, \begin{eqnarray} \|\overline{\lambda}-\overline{\lambda+d+1}\|&=&\sum_{i=1}^d|\frac{\lambda_i}{n}-\frac{\lambda_i+d+1}{n+d(d+1)}|\\ &\leq&d\cdot \max_{i=1,\ldots,d}|\frac{\lambda_i}{n}-\frac{\lambda_i+d+1}{n+d(d+1)}|\\ &=&d\cdot \max_{i=1,\ldots,d}|\frac{\lambda_id(d+1)-n(d+1)}{n(n+d(d+1))}|\\ &\leq&\frac{d(d+1)}{n^2}\cdot \max_{i=1,\ldots,d}|\lambda_id-n|\\ &\leq&\frac{d(d+1)^2}{n}\label{eqn3}\\ &\leq&\frac{(d+1)^3}{n}\\ &\leq&\frac{8d^3}{n}\\ (\mathrm{if\ }d\geq2)\qquad&\leq&\frac{d^6}{n} \end{eqnarray} and, at last, Lemma 2.7 from \cite{csiszar-koerner}: \begin{lemma}\label{lemma1}
 If, for $\mathbf A$ a finite alphabet and $p,q\in\cP(\mathbf A)$ we have $|p-q|\leq\Theta\leq1/2$, then
\begin{equation}
 |H(p)-H(q)|\leq-\Theta\log\frac{\Theta}{|\mathbf A|}.
\end{equation} \end{lemma} Combining equations (\ref{eqn1}) and (\ref{eqn2}) leads to the estimate \begin{equation}
 \dim F_\lambda\leq 2^{nH(\overline{\lambda})}\qquad (\lambda\in \bbmY_{d,n}).\label{eqn4}
\end{equation} Deriving a lower bound on $\dim F_\lambda$ is slightly more involved: Let $n\geq2d^2$. Then \begin{eqnarray}
 \dim F_\lambda&=&\frac{n!}{\prod_{i=1}^d(\lambda_i+d+1)!}\\
&=&\frac{1}{(n+d(d+1))\cdot\ldots\cdot(n+1)}\frac{(n+d(d+1))!}{\prod_{i=1}^d(\lambda_i+d+1)!}\\ &\geq&\frac{1}{(2n)^{2d^2}}\frac{(n+d(d+1))!}{\prod_{i=1}^d(\lambda_i+d+1)!}\\ &\geq&\frac{1}{(2n)^{2d^2}}\frac{1}{(2n)^d}2^{(n+d(d+1))H(\overline{\lambda+d+1})}\\ &\geq&\frac{1}{(2n)^{3d^2}}2^{nH(\overline{\lambda+d+1})}\\ &\geq&\frac{1}{(2n)^{5d^2}}2^{n(H(\overline{\lambda})+\frac{d^6}{n}\log\frac{d^5}{n})}\\ &=&2^{n(H(\overline{\lambda})+\frac{d^6}{n}\log\frac{d^5}{n}-\frac{5d^2}{n}\log(2n))}\label{eqn5}. \end{eqnarray} Set $\beta_1(n):=-\frac{d^6}{n}\log\frac{d^5}{n}+\frac{5d^2}{n}\log(2n)$, then there is an $N_1\in\nn$ such that for all $n\geq N_1$ we have \begin{align} \dim F_\lambda\geq2^{n(H(\overline\lambda)-\beta_1(n))}.\label{lower-bound-on-dimF} \end{align} An important step in the application of the representation theory of the symmetric group to quantum information theory was the following theorem: \begin{theorem}[\cite{thm2}]\label{theorem:keyl-werner} For $\lambda\in \bbmY_{d,n}$ and $\sigma\in\cs(\hr)$ ($\dim\hr=d$) with spectrum $s$ it holds \begin{align}
 \tr\{P_\lambda\sigma^{\otimes n}\}\leq(n+1)^{d(d-1)/2}2^{-nD(\overline\lambda||s)}.
\end{align}
\end{theorem}
We are now ready to prove our main theorem:
\begin{proof}[Proof of Theorem \ref{thm:thetheorem}]
Let $n\in\nn$ and $\Phi_{AB}^n$ as defined in (\ref{eqn:fundamental-definition}). In order to prove inequality (\ref{eqn-thm:thetheorem-1}) we will employ the gentle-measurement Lemma from \cite{winter} in the version given in \cite{wilde-book}. We first prove the following lower bound:
\begin{align}
\varpi(\epsilon,n)&:=\tr\{(P_\epsilon^A\otimes P_\epsilon^B)P_\epsilon^{AB}\rho_{AB}^{\otimes n}\}\\
&=1-\sum_{\mu\in\mathfrak C_\eps^n(r_A)}\sum_{\nu\in\mathfrak C_\eps^n(r_B)}\sum_{\lambda\in\mathfrak C_\eps^n(r_{AB})}\tr\{P_\mu\otimes P_\nu)P_\lambda\rho_{AB}^{\otimes n}\}\\
&\geq1-(n+1)^{d_A^4d_B^4}2^{-n\tfrac{1}{2\ln2}\epsilon^2}\\
&=1-2^{-n(\tfrac{1}{2\ln2}\epsilon^2-\beta_2(n))},
\end{align}
with $\beta_2(n):=\tfrac{4}{n}\log(d_Ad_Bn)$. Then the gentle-measurement lemma implies that
\begin{align}
\|\Phi_{AB}^n-\rho_{AB}^{\otimes n}\|_1\leq2\sqrt{\varpi(\epsilon,n)}.
\end{align}
Thus, every choice of function $\gamma$ satisfying $\gamma(\epsilon,n)\geq \tfrac{1}{2}\beta_2(n)$ for all $n\geq N(\epsilon)$ is good enough to prove inequality (\ref{eqn-thm:thetheorem-1}). We will now derive further such lower bounds on $\gamma$ and later add all these lower bounds together in order to define $\gamma$. We proceed to inequality (\ref{eqn-thm:thetheorem-4}): Let $k\in\nn$, then
\begin{align}
\tr\{\tr_B\{\Phi_{AB}^n\}^k\}&=\varpi(\epsilon,n)^k\tr\{\tr_B\{(P_\epsilon^A\otimes P_\epsilon^B)P_\epsilon^{AB}\rho_{AB}^{\otimes n}(P_\epsilon^A\otimes P_\epsilon^B)\}^k\}\\
&\leq\varpi(\epsilon,n)^k\tr\{\tr_B\{(P_\epsilon^A\otimes P_\epsilon^B)\rho_{AB}^{\otimes n}(P_\epsilon^A\otimes P_\epsilon^B)\}^k\}\\
&\leq\varpi(\epsilon,n)^k\tr\{\tr_B\{(P_\epsilon^A\otimes \eins_{\hr_B}^{\otimes n})\rho_{AB}^{\otimes n}(P_\epsilon^A\otimes \eins_{\hr_B}^{\otimes n})\}^k\}\\
&=\varpi(\epsilon,n)^k\tr\{P_\epsilon^A(\rho_{A}^k)^{\otimes n}\},
\end{align}
where the first inequality follows since $P_\eps^{AB}$ commutes both with $P_\eps^A\otimes P_\eps^B$ and with $\rho_{AB}^{\otimes n}$ and we have $\tr\{XY\}\leq\tr\{Y\}$ for all $X,Y\in\mathcal B(\kr)$ whenever $0\leq X\leq\eins_\kr$ and $0\leq Y$ hold true (and $\kr$ is a finite-dimensional Hilbert space), while the second is a consequence of the inequality $\tr_{\kr_B}\{(\eins_{\kr_A}\otimes P^B)X_{AB}(\eins_{\kr_A}\otimes P^B\}\leq\tr_B\{X_{AB}\}$, which is valid for any nonnegative operator $X_{AB}$ on a composite system $\kr_A\otimes\kr_B$ (where both $\kr_A$ and $\kr_B$ are finite-dimensional Hilbert spaces) and projection $P^B\in\mathcal B(\kr_B)$.\\
It makes sense to treat the term $\tr\{P_\epsilon^A(\rho_{A}^k)^{\otimes n}\}$ separately: Let $Q^A_t$ be the frequency-typical subspaces corresponding to $\rho_A$, as defined in subsection \ref{subsec:Reformulation of the conjecture and initial approaches}. It is known that $P_\lambda Q_t=0$ whenever $t$ is not majorized by $\lambda$ (see \cite[Lemma 1.21]{christandl} and also the proof of Theorem 2.14 therein), so that
\begin{align}
\tr\{P_\epsilon^A(\rho_{A}^k)^{\otimes n}\}&=\varpi(\epsilon,n)^k\sum_t\prod_{i=1}^nr_A(i)^{k\cdot f(i)}\tr\{P_\epsilon Q^A_t\}\label{eqn:basic-estimates-start}\\
&=\varpi(\epsilon,n)^k\sum_t\prod_{i=1}^nr_A(i)^{k\cdot f(i)}\tr\{P_\epsilon Q^A_t\}\\
&=\varpi(\epsilon,n)^k\sum_{\mu\in\mathfrak C_\eps^n(r_A)}\sum_t\prod_{i=1}^nr_A(i)^{k\cdot f(i)}\tr\{P_\mu Q^A_t\}\\
&=\varpi(\epsilon,n)^k\sum_{\mu\in\mathfrak C_\eps^n(r_A)}\sum_{t\preceq\lambda}\prod_{i=1}^nr_A(i)^{k\cdot f(i)}\tr\{P_\mu Q^A_t\}\\
&\leq\varpi(\epsilon,n)^k(n+1)^{2d_A^2}\sum_{\mu\in\mathfrak C_\eps^n(r_A)\|\leq\epsilon}\sum_{t\preceq\lambda}\prod_{i=1}^nr_A(i)^{k\cdot f(i)}\dim(F_\mu)\\
&\leq\varpi(\epsilon,n)^k(n+1)^{2d_A^2}\sum_{\mu\in\mathfrak C_\eps^n(r_A)}\sum_{t\preceq\lambda}\prod_{i=1}^nr_A(i)^{k\cdot \lambda(i)}\dim(F_\mu)\\
&=\varpi(\epsilon,n)^k(n+1)^{2d_A^2}\sum_{\mu\in\mathfrak C_\eps^n(r_A)}\sum_{t\preceq\lambda}2^{n\cdot(\sum_{i=1}^{d_A}k\cdot \bar\lambda_i\log r_A(i)+H(\bar\mu))}.\label{eqn:basic-estimates-final}
\end{align}
Given any state $\rho\in\cs(\kr)$ for some Hilbert space $\kr$, let $s_{\min}(\rho)$ denote its smallest nonzero eigenvalue. Set $c':=\min\{s_{\min}(\rho_A),s_{\min}(\rho_B),s_{\min}(\rho_{AB})\}$. Then any of the terms in above sum can be upper bounded as follows:
\begin{align}\label{eq:entropy-estimate}
\sum_{i=1}^{d_A}k\cdot \bar\lambda_i\log r_A(i)+H(\bar\mu)&\leq (1-k)\cdot H(r_A)+k\cdot\eps\cdot(|\log c'|+\log(\eps\cdot d^{-1})),
\end{align}
where $\|\bar\lambda-r_A\|_1\leq\eps$ was used in combination with Lemma \ref{lemma1}. It follows that
\begin{align}\label{eqn:s_min}
\tr\{P_\epsilon^A(\rho_{A}^k)^{\otimes n}\}&\leq\varpi(\epsilon,n)^k(n+1)^{6d_A^2}2^{-n\cdot((k-1)H(r_A)+k\cdot\eps\cdot(|\log c'|+\log(\eps\cdot d^{-1}))},
\end{align}
and with the appropriate and obvious choice of $\gamma(\eps,n)$ (by assumption, $k\geq2$ so that $1\leq k/(k-1)\leq2$ holds. Thus $\gamma$ depends on $k$ only through the term $\frac{1}{n}\log\varpi(\epsilon,n)^k$) the claim follows.\\
The discussion can now be continued in the same manner to derive the estimate (\ref{eqn-thm:thetheorem-9}), where again the function $\gamma$ gets updated such that it gives an upper bound on the sum of all its predecessors.\\
Finally, the inequality (\ref{eqn-thm:thetheorem-1}) is the easiest to prove since it only requires one to verify the estimate
\begin{align}
\tr\{(\Phi_{AB}^n)^k\}\leq\varpi(\epsilon,n)^k\tr\{P_\eps^{AB}(\rho_{AB}^k)^{\otimes n}\},
\end{align}
which is a consequence of the inequality $XYX^\dag\leq XX^\dag$ that holds true whenever $0\leq Y\leq\eins$. After using a reasoning along the lines of inequalities (\ref{eqn:basic-estimates-start}) until (\ref{eqn:basic-estimates-final}) for $\rho_{AB}$ instead of $\rho_A$, one proceeds with the inequality (\ref{eq:entropy-estimate}) and uses (\ref{eqn:s_min}) where one sets $r_{AB}$ in place of $r_A$. The proof is finally finished by adding all the sub-exponential correction terms to form the function $\gamma$.
\end{proof}
\begin{proof}[Proof of Theorem \ref{thm:additionaltheorem}] While it may seem that achieving a lower bound like the one we are aiming at is a trivial thing, this is in fact not the case here due to the multiparty nature of the problem. Clearly, $\lambda\in\mathfrak C_\eps^n(r_{AB})$ implies that $\lambda\approx n\cdot r_{AB}$ for at least one $\lambda$ once $n$ is large enough. However, there is an additional constraint on those representations of $S_n$ that appear in the support of $P_\lambda^{AB}(P_\mu^A\otimes P_\nu^B)$, and for that reason the proof becomes a nontrivial extension of what is known already. The obvious approach would certainly be to deduce that there are projections $P_{\lambda,i}^{AB}$ ($i=1,2,\ldots,m$ for some number $m$ that may be strictly larger than one) such that $\Sigma:=\sum_iP_{\lambda,i}^{AB}= P_\lambda^{AB}(P_\mu^A\otimes P_\nu^B)$. However, it is not clear that $\Sigma\rho_{AB}^{\otimes n}=\rho_{AB}^{\otimes n}\Sigma$. Also, bounds like $(P_{\lambda,1}+P_{\lambda,2})\rho_{AB}^{\otimes n}(P_{\lambda,1}+P_{\lambda,2})\geq P_{\lambda,1}\rho_{AB}^{\otimes n}P_{\lambda,1}+P_{\lambda,2}\rho_{AB}^{\otimes n}P_{\lambda,2}$ are not valid in general (this pinching inequality actually holds in the reverse direction, with equality holding for example if $P_{\lambda,1}\rho_{AB}^{\otimes n}P_{\lambda,2}=0$) and therefore calculation of the lower bound becomes less straightforward than expected.\\
Our route to approach this problem is to first derive bounds on quantities $\tr\{P_{\lambda,i_1}\rho_{AB}^{\otimes n}\cdot\ldots\cdot P_{\lambda,i_k}\rho_{AB}^{\otimes n}\}$ when one of the $P_{\lambda,i}\leq P_\mu\otimes P_\nu$ for a pair of Young frames $(\mu,\nu)$ not being close to the pair $(r_A,r_B)$. This is the content of the following Lemma:
\begin{lemma}\label{lemma-2} Let $\rho_{AB}\in\cs(\hr_A\otimes\hr_B)$ be a quantum state with spectrum $r_{AB}$ and marginals having spectra $r_A$ and $r_B$. Let $\mu\in \bbmY_{d_A,n}$, $\nu\in \bbmY_{d_B,n}$, $\lambda\in \bbmY_{d_Bd_A,n}$ be Young frames. Let $k\in\nn$.

If $\|\bar\mu-r_A\|_1>\delta$, then for $P_{\lambda,i_1},\ldots,P_{\lambda,i_k}$ ($i_1,\ldots,i_k\in[m_\lambda^{AB}])$ with at least one of the indices (let this be $i_x$) satisfying
$P_{\lambda,i_x}\leq P_\mu\otimes P_\nu$ and the others $P_{\lambda,i'}\leq P_{\mu'}\otimes P_{\nu'}$ for arbitrary other Young frames $\mu'\in\bbmY_{d_A,n}$, $\nu'\in\bbmY_{d_B,n}$ it holds
\begin{align}
|\tr\{P_{\lambda,i_1}\rho_{AB}^{\otimes n}\cdot\ldots\cdot P_{\lambda,i_k}\rho_{AB}^{\otimes n}\}|\leq2^{-nc\delta^2}2^{-n(k-1)(D(\bar\lambda\|r_{AB})+H(\bar\lambda)-\beta_3(n))}.
\end{align}
The function $\beta_3$ is given by $\beta_3(n):=\frac{(d_Ad_B)^2}{n}\log(2n)+\beta_1(n)$.
\end{lemma}
\begin{remark} The lemma can w.l.o.g. be read with the roles of $A$ and $B$ interchanged.
\end{remark}
\begin{proof}[Proof of Lemma \ref{lemma-2}] We consider the first statement first. Let us take a look at $P_\lambda\rho_{AB}^{\otimes n}$ first. Observe that the two operators in this product commute. Since $P_\lambda\rho_{AB}^{\otimes n}$ is invariant under permutations, we can write it as \begin{align} P_\lambda\rho_{AB}^{\otimes n}=\sum_{i,j=1}^{m_{\lambda}}c_{ij}Y_{ij}, \end{align} where the operators $Y_{ij}\in\mathcal B(\hr_{AB}^{\otimes l})$ satisfy \begin{align} P_{\lambda,i}Y_{ij}P_{\lambda,j}=Y_{ij},\qquad Y_{ij}Y_{kl}=\delta(j,k)Y_{il}\qquad\mathrm{and}\qquad Y_{jj}=P_{\lambda,j}. \end{align} Since $P_\lambda\rho_{AB}^{\otimes n}$ is self-adjoint, we get \begin{align}
 \sum_{i,j=1}^{m_{\lambda}}c_{ij}Y_{ij}&=\sum_{i,j=1}^{m_{\lambda}}\bar c_{ij} Y_{ij}^\dag\\
&=\sum_{i,j=1}^{m_{\lambda}}\bar c_{ij} Y_{ji}, \end{align} from which it follows that $\bar c_{ij}=c_{ji}$. Also, for every $i,j\in[m^{AB}_\lambda]$ we know that \begin{align} c_{ii}Y_{ii}+c_{ij}Y_{ij}+c_{ji}Y_{ji}+c_{jj}Y_{jj}=(P_{\lambda,i}+P_{\lambda,j})\rho_{AB}^{\otimes n}(P_{\lambda,i}+P_{\lambda,j})\geq0. \end{align} By choosing appropriate bases for $\supp(P_{\lambda,i})$ and $\supp(P_{\lambda,j})$, this translates to the statement \begin{align} \left(\begin{array}{ll} c_{ii} & c_{ij}\\ c_{ji} & c_{jj}\end{array}\right)\geq 0. \end{align} This now shows us that $|c_{ij}|^2\leq|c_{ii}|\cdot|c_{jj}|$ has to hold and that all the $c_{ii}$, $i=1,\ldots,[m^{AB}_\lambda]$, are nonnegative real numbers. We now prove the promised inequality:
\begin{align}
|\tr\{P_{\lambda,i_1}\rho_{AB}^{\otimes n}\cdot\ldots\cdot P_{\lambda,i_k}\rho_{AB}^{\otimes n}\}|
&=|\sum_{j_1,l_1=1}^{m^{AB}_\lambda}c_{j_1l_1}\cdot\ldots\cdot\sum_{j_k,l_k=1}^{m^{AB}_\lambda}c_{j_kl_k}\tr\{P_{\lambda,i_1}Y_{j_1l_1}\cdot\ldots\cdot P_{\lambda,i_k}Y_{j_kl_k}\}|\\
&=|\prod_{j=1}^kc_{i_ji_{j+1}}\dim(F_\lambda)|\label{eqn:k-state-start}\\
&=\prod_{j=1}^k|c_{i_ji_{j+1}}|\dim(F_\lambda)\\
&\leq\prod_{j=1}^k\sqrt{|c_{i_ji_{j}}|\cdot|c_{i_{j+1}i_{j+1}}|}\dim(F_\lambda)\\
&=\prod_{j=1}^k|c_{i_ji_{j}}|\dim(F_\lambda).\label{eqn:k-state-end}
\end{align}
Observe that
\begin{align} |c_{ii}|&=c_{ii}\\ &=\tr\{Y_{ii}\sum_{k,l}c_{kl}Y_{kl}\}/\dim F_\lambda\\ &=\tr\{P_{\lambda,i}P_\lambda\rho_{AB}^{\otimes n}\}/\dim F_\lambda\\
&=\tr\{P_{\lambda,i}\rho_{AB}^{\otimes n}\}/\dim F_\lambda,
\end{align}
so together with equations (\ref{eqn:k-state-start}) to (\ref{eqn:k-state-end}) we can combine this to get
\begin{align}
|\tr\{P_{\lambda,i_1}\rho_{AB}^{\otimes n}\cdot\ldots\cdot P_{\lambda,i_k}\rho_{AB}^{\otimes n}\}|&= \dim(F_\lambda)^{-k}\tr\{P_{\lambda,i_1}\rho_{AB}^{\otimes n}\}\cdot\ldots\cdot\tr\{P_{\lambda,i_k}\rho_{AB}^{\otimes n}\}\dim(F_\lambda)\label{eqn6}\\
(\mathrm{by\ assumption})\qquad&\leq\tr\{(P_\mu\otimes P_\nu)\rho_{AB}^{\otimes n}\}\tr\{P_\lambda\rho_{AB}^{\otimes n}\}^{k-1}\dim(F_\lambda)^{1-k}\\
(\mathrm{since\ }P_\nu\leq\eins_{\hr_{B}^{\otimes n}})\qquad&\leq\tr\{P_\mu\rho_A^{\otimes n}\}\tr\{P_\lambda\rho_{AB}^{\otimes n}\}^{k-1}\dim(F_\lambda)^{1-k}\\
(\mathrm{Theorem\ \ref{theorem:keyl-werner},\ inequality\ (\ref{lower-bound-on-dimF})})\ \ &\leq(2n)^{2kd_A^2d_B^2}2^{-nD(\bar\mu||r_A)}2^{-n(k-1)(D(\bar\lambda\|r_{AB})+H(\bar\lambda)-\beta_1(n))}\\
(\mathrm{Pinsker's\ inequality})\qquad&\leq(2n)^{2kd_A^2d_B^2}2^{-nc\delta^2}2^{-n(k-1)(D(\bar\lambda\|r_{AB})+H(\bar\lambda)-\beta_1(n))}\\
&\leq2^{-nc\delta^2}2^{-n(k-1)(D(\bar\lambda\|r_{AB})+H(\bar\lambda)-\beta_3(n))}.
\end{align}
\end{proof}
\begin{lemma}\label{lemma-3} For a type $N(\cdot)$ on $[d_Ad_B]^n$ and its corresponding typeclass $T_N\subset[d_Ad_B]^n$ and $e_1,\ldots,e_{d_Ad_B}$ a basis in which $\rho_{AB}$ is diagonal, let $\hr_N:=\textup{span}(\{e_{x_1}\otimes\ldots\otimes e_{x_n}:x^n\in T_N\})$. Denote the projection onto $\hr_N$ by $p_{\hr_N}$.\\
Then $\hr_N$ is invariant under the action $\mathbb B^{AB}$ of $S_n$ and for every $\lambda$ with $\lambda=N^\downarrow$, $P_\lambda P_N\neq0$.
\end{lemma}
\begin{proof} To a given $N(\cdot)$, take $T$ to be the standard tableaux for $\lambda=N^\downarrow$ which has entries $T_{1i}=i,\ T_{2i}=\lambda_1+i$ and so on, until finally $T_{\lambda_{d_Ad_B}i}=\lambda_1+\ldots\lambda_{n-1}+i$. Let $v=\otimes_{i=1}^{d_Ad_B}e_i^{\otimes N(i)}$. Denote the set of row permutations belonging to $T$ by $R_T$, the column permutations by $C_T$ and set $E_T:=\{\pi\circ\tau:\pi\in C_T,\tau\in R_T\}$. Note that $V_\lambda:=\textrm{span}(\{E(T)v:v\in\hr^{\otimes n},\ T-\textrm{standard\ tableaux\ for\ }\lambda\})$ is the isotypical vectorspace belonging to $\lambda$ - it holds $\supp(P_\lambda)=V_\lambda$.\\ We calculate the overlap of $v$ with a suitably chosen element of $V_\lambda$: \begin{align} \langle v,\mathbb B(E_T)v\rangle&=\sum_{\pi\in C_T}\sgn(\pi)\sum_{\tau\in R_T}\langle \mathbb B(\pi)\mathbb B(\tau)v,v\rangle\\ &=|R_T|\sum_{\pi\in C_T}\sgn(\pi)\langle \mathbb B(\pi)v,v\rangle\\ &=|R_T|, \end{align} since $\langle \mathbb B(\pi)v,v\rangle=0$ for every $C_T\ni\pi\neq e$. Now assume that $\hr_t$ contains no irreducible subspace corresponding to $\lambda$. Then, of course, for every vector $w\in\supp(P_\lambda)$ we have $w\perp\hr_t$. But by the preceding, the vector $w:=\mathbb B(E_T)v\in\supp(P_\lambda)$ is not perpendicular to $\hr_t$.\\ Thus, there must be at least one copy of $F_\lambda$ in $\hr_t$, which is what we set out to prove. \end{proof}
\begin{lemma}\label{lemma-4} For any $\lambda\in \bbmY_{d_Ad_B,n}$ and $k\in\nn$, it holds that
\begin{align}
\tr\{P_\lambda(\rho_{AB}^k)^{\otimes n}\}\geq2^{-nD(\bar\lambda\|r_{AB})}2^{-n(k-1)(D(\bar\lambda\|r_{AB})+H(\bar\lambda)+\beta_1(n))}2^{n(2-k)H(\bar\lambda)}.
\end{align}
\end{lemma}
Note that, if $D(\bar\lambda\|r_{AB})=\infty$, the right hand side of above inequality equals zero.
\begin{proof}[Proof of Lemma \ref{lemma-4}] By Lemma \ref{lemma-3}, for the subspace $\hr_t$ defined by the typeclass corresponding to $\lambda$, we have $P_{\lambda,i}\leq p_{\hr_t}$ for at least one $i\in[m^{AB}_\lambda]$. Also, $\langle v,(\rho_{AB}^k)^{\otimes n}v\rangle =\prod_{j=1}^{d_Ad_B}r_j^{kt(j)}$ for every $v\in\hr_t$, and moreover it holds $p_{\hr_t}\rho_{AB}^{\otimes n}=\rho_{AB}^{\otimes n}p_{\hr_t}$. Thus,
\begin{align}
\tr\{(P_\lambda\rho_{AB}^{\otimes n})^k\}&=\tr\{P_\lambda(\rho_{AB}^k)^{\otimes n}\}\\
&\geq \tr\{P_{\lambda}p_{\hr_t}(\rho_{AB}^k)^{\otimes n}\}\\
&=\prod_{j=1}^{d_Ad_B}r_{AB}(j)^{k\lambda_j}\tr\{P_{\lambda}p_{\hr_t}\}\\
&\geq2^{nk\sum_{j=1}^{d_Ad_B}\overline\lambda_j\log r_{AB}(j)}\dim(F_\lambda)\\
&=2^{nk\sum_{j=1}^{d_Ad_B}\overline\lambda_j\log r_{AB}(j)}\dim(F_\lambda)^{k}\dim(F_\lambda)^{1-k}\\
&\geq2^{-nD(\bar\lambda\|r_{AB})}2^{-n(k-1)(D(\bar\lambda\|r_{AB})+H(\bar\lambda)+\beta_1(n))}.
\end{align}
\end{proof}
We are now finally coming to the derivation of the lower bound (\ref{eqn:additionaltheorem-1}) on $\tr\{(\Phi_{AB}^n)^k\}$. Our approach is to compare quantities of the form $\tr\{(P_\lambda)\rho_{AB}^{\otimes n})^k\}$ for which we know a lower bound from Lemma \ref{lemma-4} with quantities of the form $\tr\{((P_\mu\otimes P_\nu)P_\lambda\rho_{AB}^{\otimes n}(P_\mu\otimes P_\nu))^k\}$, for which a lower bound seems hard to get at least at first sight.\\
Let $\lambda\in\mathfrak C_\eps^n(r_{AB})$, and let $E\subset[m_\lambda^{AB}]$ denote the set of indices such that for all $i\in E$ we have $P_{\lambda,i}(P_\mu\otimes P_\nu)=0$ whenever $\|r_A-\bar\mu\|_1\leq\eps$ or $\|r_B-\bar\nu\|_1\leq\eps$. We further define $D:=[m_\lambda^{AB}]\backslash E$, the complement of $E$ within $[m_\lambda^{AB}]$. It then holds that $P_\lambda(P_\eps^A\otimes P_\eps^B)=\sum_{i\in D}P_{\lambda,i}$. Define $P_D:=\sum_{i\in D}P_{\lambda,i}$ and $P_E:=\sum_{i\in E}P_{\lambda,i}$ and note that $P_D+P_E=P_\lambda$. Then, with $\bX:=\{E,D\}$ an alphabet we can define for each $\lambda\in\bbmY_{d_Ad_B,n}$ a function $f:\bX^k\to\mathbb C$ by
\begin{align}
f_\lambda(x^k):=\tr\{P_{x_1}\rho_{AB}^{\otimes n}\cdot\ldots\cdot P_{x_k}\rho_{AB}^{\otimes n}\}.
\end{align}
We are intersted in the derivation of a lower bound on the function value
\begin{align}
f_\lambda((D,\ldots,D))=\tr\{((P_\eps^A\otimes P_\eps^B)P_\lambda^{AB}\rho_{AB}^{\otimes n}(P_\eps^A\otimes P_\eps^B))^k\}.
\end{align}
In order to derive a lower bound on this quantity we write it as
\begin{align}
f_\lambda((D,\ldots,D))&=\sum_{x^k}f(x^k)-\sum_{x^k\neq(D,\ldots,D)}f(x^k)\\
&=\tr\{P_\lambda^{AB}(\rho_{AB}^k)^{\otimes n}\}-\sum_{x^k\neq(D,\ldots,D)}f(x^k)\\
&\geq\tr\{P_\lambda^{AB}(\rho_{AB}^k)^{\otimes n}\}-|\bX|^k\max_{x^k\neq(D,\ldots,D)}|f(x^k)|.
\end{align}
Pick any $x^k\neq(D,\ldots,D)$. Without loss of generality it holds $x_1=E$. In that case, we can write
\begin{align}
f_\lambda(x^k)&=|\tr\{P_{x_1}\rho_{AB}^{\otimes n}\cdot\ldots\cdot P_{x_k}\rho_{AB}^{\otimes n}\}|\\
&\leq (m_\lambda^{AB})^k|\tr\{P_{\lambda,i_1}\rho_{AB}^{\otimes n}\cdot\ldots\cdot P_{\lambda,i_k}\rho_{AB}^{\otimes n}\}|,
\end{align}
where $i_1,\ldots,i_k$ obey $i_1\in E$ and $i_2,\ldots,i_k$ are taken either from $E$ or from $D$. Now, Lemma \ref{lemma-2} can be directly applied - with $P_{\lambda,i_1}\leq P_\mu^A\otimes P_\nu^A$ for some pair $\mu,\nu$ where at least $\|\bar\mu-r_A\|\geq\eps$ or $\|\bar\nu-r_B\|\geq\eps$. This yields
\begin{align}
f_\lambda(x^k)\leq2^{-nc\eps^2}2^{-n(k-1)(D(\bar\lambda\|r_{AB})+H(\bar\lambda)-\beta_3(n))}.
\end{align}
Since on the other hand we know from Lemma \ref{lemma-4} that
\begin{align}
\sum_{x^k\in\bX^k}f_\lambda(x^k)\geq2^{-nD(\bar\lambda\|r_{AB})}2^{-n(k-1)(D(\bar\lambda\|r_{AB})+H(\bar\lambda)+\beta_1(n))}2^{n(2-k)H(\bar\lambda)},
\end{align}
we can conclude that
\begin{align}
f_\lambda((D,\ldots,D))&\geq2^{-nD(\bar\lambda\|r_{AB})}2^{-n(k-1)(D(\bar\lambda\|r_{AB})+H(\bar\lambda)+\beta_1(n))}\\
&\qquad
-2^k2^{-nc\delta^2}2^{-n(k-1)(D(\bar\lambda\|r_{AB})+H(\bar\lambda)-\beta_3(n))}\\
&=2^{-n(k-1)(D(\bar\lambda\|r_{AB})+H(\bar\lambda)-\beta_3(n))}\left(2^{-nD(\bar\lambda\|r_{AB})}-2^k2^{-nc\epsilon^2}\right).
\end{align}
We may now take any sequence $(\lambda^{(n)}_{n\in\nn}$ of Young frames converging to $r_{AB}$ such that $r_{AB}(j)=0$ implies $\lambda_j^{(n)}=0$ for all $j=1,\ldots,d_Ad_B$ and set
\begin{align}
\nu_k(\eps,n)&:=D(\bar\lambda^{(n)}\|r_{AB})+|H(r_{AB})-H(\bar\lambda^{(n)})|\\
&\qquad-\beta_3(n)+\frac{1}{n(k-1)}\log\left(\varsigma(\eps,n)^k(2^{-nD(\bar\lambda^{(n)}\|r_{AB})}-2^k2^{-nc\epsilon^2})\right).
\end{align}
It is clear from e.g. Lemma \ref{lemma1} that $\lim_{n\to\infty}\nu_k(\eps,n)=0$ holds for all $\eps>0$ and $k\geq2$. Since $\rho_{AB}^{\otimes n}$ is permutation-invariant it holds for all choices of $\xi^{(1)},\ldots,\xi^{(k)}\in\bbmY_{d_Ad_B,n}$ and $i_1\in[m_{\xi^{(1)}}^{AB}],\ldots,i_k\in[m_{\xi^{(k)}}^{AB}]$ that
\begin{align}
\tr\{P_{\xi^{(1)},i_1}\rho_{AB}^{\otimes n}\cdot\ldots\cdot P_{\xi^{(k)},i_k}\rho_{AB}^{\otimes n}\}=0
\end{align}
whenever there are $a,b\in[k]$ such that $\xi_a\neq\xi_b$. Thus for all $n\in\nn$ we have
\begin{align}
\tr\{(\Phi_{AB}^n)^k\}&=\sum_{\lambda\in\mathfrak C_{\eps,n}(r_{AB})}f_\lambda((D,\ldots,D))\\
&\geq f_{\lambda^{(n)}}((D,\ldots,D))\\
&\geq2^{-n(k-1)(H(r_{AB})-\nu_k(\eps,n))}.
\end{align}
\end{proof}
\end{section}
\emph{Acknowledgements.} Janis N\"otzel wants to thank Gisbert Jan\ss en for pointing out to him the importance of the problem, Igor Bjelakovic for his encouragement and Matthias Christandl for a helpful discussion. Many thanks go to Igor Bjelakovic, Holger Boche and Gisbert Janssen for weekly discussions about representation theory.\\ This work was supported by the BMBF via grant 01BQ1050, by the DFG via grant NO 1129/1-1 and also by the ERC Advanced Grant IRQUAT, the Spanish MINECO Project No. FIS2013-40627-P  and the Generalitat de Catalunya CIRIT Project No. 2014 SGR 966.


\begin{thebibliography}{99}
\bibitem{thm2} Theorem \ref{theorem:keyl-werner} seems to have a long history, including independent rediscoveries. This history is explained in more detail in \cite{christandl}. According to \cite{christandl}, it first appeared in \cite{rudnicki-alicki-sadowsky}, was then independently proven in \cite{keyl-werner}, later appeared in \cite{hayashi-matsumoto} with a shortened proof, and was, at last, restatet in \cite{christandl-mitchison}, again with the (obtainded independently from \cite{hayashi-matsumoto}) shortened version of the proof.
\end{thebibliography}
\end{document}